\documentclass[conference, letterpaper]{IEEEtran}
\usepackage{amsmath,amssymb}
\usepackage{subfigure}
\usepackage{graphicx,graphics,color,psfrag}
\usepackage{cite,balance}
\usepackage{caption}
\allowdisplaybreaks
\usepackage{algorithm}
\usepackage{algorithmic}
\usepackage{accents}
\usepackage{amsthm}
\usepackage{bm}
\usepackage{url}
\usepackage[english]{babel}
\usepackage{multirow}
\usepackage{enumerate}
\usepackage{cases}
\usepackage{stfloats}
\usepackage{dsfont}
\usepackage{color,soul}
\usepackage{amsfonts}
\usepackage{cite,graphicx,amsmath,amssymb}
\usepackage{subfigure}
\usepackage{fancyhdr}
\usepackage{hhline}
\usepackage{graphicx,graphics}
\usepackage{array,color}
\usepackage{mathtools}
\usepackage{amsmath}

\newtheorem{definition}{\emph{\underline{Definition}}}

\newtheorem{lemma}{\emph{\underline{Lemma}}}

\newtheorem{remark}{\bf \emph{\underline{Remark}}}

\def\({\left(}
\def\){\right)}

\def\b0{{\mathbf{0}}}

\newcommand{\nn}{\nonumber}

\usepackage[top=0.65in, bottom=1.06in, left=0.58in, right=0.58in]{geometry}
\begin{document}
\captionsetup[figure]{name={Fig.}}

\title{\huge Joint Beam
	Scheduling and Power Allocation for SWIPT in Mixed Near- and Far-Field Channels \vspace{-0.5cm} 
} 
\author{\IEEEauthorblockN{Yunpu Zhang$^{*}$, Changsheng You$^{*}$, Weijie Yuan$^{*}$, Fan Liu$^{*}$, and Rui Zhang$^{\dagger\ddagger}$}
	\IEEEauthorblockA{$^{*}$Department of Electronic and Electrical Engineering, 
		Southern University of Science and Technologyy, Shenzhen, China\\
		$^{\dagger}$The Chinese University of Hong Kong (Shenzhen), and
	Shenzhen Research Institute of Big Data, Shenzhen, China\\
			$^{\ddagger}$Department of Electrical and
		Computer Engineering, National University of Singapore, Singapore\\
		Email: zhangyp2022@mail.sustech.edu.cn; \{youcs, yuanwj, liuf6\}@sustech.edu.cn; rzhang@cuhk.edu.cn}\vspace{-1cm}}

\maketitle

\begin{abstract} 
Extremely large-scale array (XL-array) has emerged as a promising technology to enhance the spectrum efficiency and spatial resolution in future wireless networks, leading to a fundamental paradigm shift from conventional far-field communications towards the new near-field communications. Different from the existing works that mostly considered simultaneous wireless information
and power transfer (SWIPT) in the far field, we consider in this paper a new and practical scenario, called \emph{mixed near- and far-field} SWIPT, in which energy harvesting (EH) and information decoding (ID) receivers are located in the near- and far-field regions of the XL-array base station (BS), respectively.
Specifically, we formulate an optimization problem to maximize the weighted sum-power harvested at all EH receivers by jointly designing the BS beam scheduling and power allocation, under the constraints on the ID sum-rate and BS transmit power.
%In contrast to the existing works that mostly focus on the near-field communications.
%	we consider in this paper a new and practical scenario, named \emph{mixed} near- and far-field communications.
%In this paper, we consider SWIPT in mixed near- and far-channels, where a base station (BS) equipped with extremely large-scale array (XL-array) simultaneously serve multiple near-field energy harvesting (EH) receivers and far-field information decoding (ID) receivers.
To solve this non-convex optimization problem, an efficient algorithm is proposed to obtain a suboptimal solution by leveraging the binary variable elimination and successive convex approximation methods.
Numerical results demonstrate that our proposed joint design achieves substantial performance gain over other benchmark schemes.
\end{abstract}
\vspace{-0.3cm}
\section{Introduction}
\vspace{-0.2cm}
Extremely large-scale array/surface (XL-array/surface) has been  envisioned as a promising technology to achieve super-high spectrum efficiency and spatial resolution for future sixth-generation (6G) wireless networks \cite{8869705}. With the significant increase in number of antennas, the well-known \emph{Rayleigh distance} will expand to dozens or even hundreds
of meters. This thus leads to a fundamental paradigm shift in the electromagnetic (EM) field characteristics, from the conventional far-field communications towards the new \emph{near-field communications} \cite{liu2023near,9903389}.

While most of the existing works have considered either the near- or far-field communications, the \emph{mixed} near- and far-field communications are likely to appear, in which
there exist both near- and far-field users in the network \cite{zhang2023mixed}.
%For instance, consider a base station (BS) that has an 
%antenna of diameter $0.5$ meter (m) and operates at a frequency of $30$ GHz. The well-known Rayleigh
%distance in this scenario is about $50$ m. 
This means that in a typical communication
scenario, the users may be located in both the near- and far-field regions from the base station (BS), hence causing more complicated interference issue. Specifically, an interesting observation was unraveled in \cite{zhang2023mixed} that due to the energy-spread effect, the near-field user may suffer strong interference from the discrete Fourier transform (DFT)-based far-field beam,
when its spatial angle is in the neighborhood of the far-field user angle. On the other hand, such power leakage from the DFT-based far-field beam can also be utilized to benefit the near-field user, leading to the new application of mixed-field simultaneous wireless information and power transfer (SWIPT) where energy harvesting (EH) and information decoding (ID) receivers are located in the near- and far-field, respectively. However, the design of mixed-field SWIPT also confronts with new challenges. In particular, thanks to the near-field beam-focusing property \cite{9738442}, the beamforming for the near-field EH receivers should be designed to maximize the EH efficiency, while at the same time, minimizing the interference to the far-field ID receivers. Besides, the beamforming design for the far-field ID receivers should take into account the energy-spread effect, which can opportunistically charge the near-field EH receivers when they are located in similar angles. Moreover, the power allocation of the BS should be carefully designed to strike the new \emph{near-and-far} beamforming tradeoff in the mixed SWIPT, with the effects of both beam focusing and energy spread taken into account.
\begin{figure}[t]
	\centering
	\includegraphics[width=6.5cm]{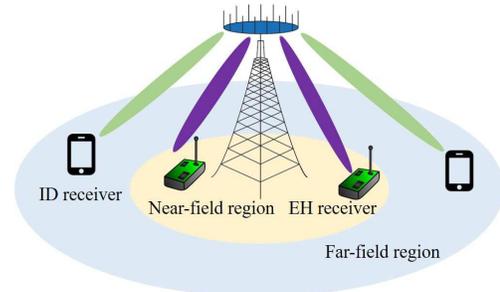}
	\caption{A mixed near- and far-field SWIPT system.}\label{fig:SM}
	\vspace{-18pt}
\end{figure}

%{\color{blue}However, when migrating to the simultaneous wireless information and power transfer (SWIPT) \cite{9669263}, there arise from new design issues: 1) the SWIPT in mixed-field communications is much more challenging than those of far-field or near-field, since it needs to strike a new trade-off between the energy harvesting (EH) receivers and information decoding (ID) receivers; 2) there remains fairness issue due to the more favorable channel conditions for near-field users compared to far-field users. 
%The above challenging issues which, however, has not been studied in the existing literature. In
%particular, it remains unknown how to design beam scheduling and power allocation in the mixed-field SWIPT  and whether the newly introduced distance-domain resolution can bring significant performance gain for the mixed-field SWIPT.}

%the induced interference from the far-field user can in turn charge near-field user, while the near-field user causes very limited interference to the far-field user due to the very limited energy leakage of near-field user and more severe path-loss of far-field user, hence arsing from new design issues for SWIPT in the mixed-field communications. Unlike the conventional far-field SWIPT, there

To address the above issues, we consider in this paper a new \emph{mixed near- and far-field} SWIPT system as shown in Fig. 1, where the BS equipped with an XL-array simultaneously serves multiple EH and ID receivers, which are located in the near- and far-field regions of the XL-array, respectively. Specifically, our goal is to maximize the weighted sum-power harvested at all EH receivers while ensuring the sum-rate for ID receivers. To this end, 
we propose a joint BS beam scheduling and power allocation design by utilizing the binary variable elimination and successive convex approximation (SCA) methods. Numerical
results demonstrate the effectiveness of the proposed scheme as compared to other benchmark schemes. In particular, we show the necessity of transmit power allocation to the EH beam in the mixed-field SWIPT, which is in
sharp contrast to the conventional far-field SWIPT case, for which all transmit power should be allocated to ID
beams.
\vspace{-0.2cm}
\section{System Model and Problem Formulation}\label{SMandPF}
\vspace{-0.1cm}
\subsection{System Model}
\vspace{-0.1cm}
We consider a mixed-field SWIPT system as shown in Fig.~\ref{fig:SM}, where a BS equipped with an $N$-antenna XL-array simultaneously serves multiple EH and ID receivers. Specifically, 
$K$ single-antenna EH receivers, denoted by $\mathcal{K}=\{1,2,\cdots,K\}$, are located in the near-field region of the XL-array for enabling efficient energy harvesting, while $M$ single-antenna ID receivers, denoted by $\mathcal{M}=\{1,2,\cdots,M\}$, are assumed to lie in the far-field region of the XL-array with the BS-ID receiver distances larger than the so-called Rayleigh distance, defined as $Z=\frac{2D^2}{\lambda}$ with $D$ and $\lambda$ denoting the antenna array aperture and carrier wavelength, respectively. {For the XL-array BS, it applies the hybrid beamforming architecture to serve $(K+M)$ EH and ID receivers with $N_{\rm RF}$ radio frequency (RF)
	chains, where $N_{\rm RF}\ge K+M$.}

%Thus, the channel characterizations of EUs and IUs with different forms are introduced in the sequel, respectively.
\subsubsection{Near- and Far-Field Channel Models}
In the following, we introduce the channel models
for the near-field EH receivers and far-field ID receivers, respectively.
\begin{figure}[t]
	\centering
	\includegraphics[width=8.5cm]{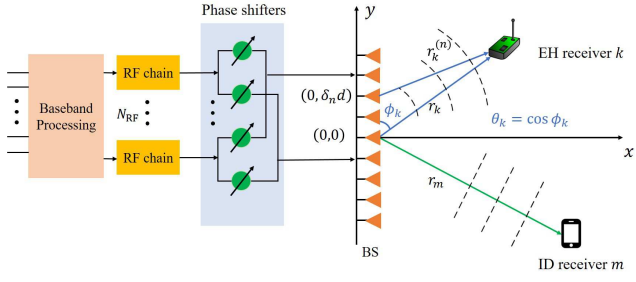}
	\caption{Illustration of the near and far-field channel models.}\label{fig:An}
	\vspace{-18pt}
\end{figure}

\underline{\bf Near-field channel model for EH receivers:}
For the near-field EH receiver $k$, its channel from the XL-array BS can be modeled as
	\vspace{-0.2cm}
\begin{equation}\label{Eq:mp}
	\mathbf{h}^{\rm EH}_{k}=	\mathbf{h}^{\rm EH}_{{\rm LoS}, k}+	\sum_{\ell=1}^{L_{k}}\mathbf{h}^{\rm EH}_{{\rm NLoS}, k, \ell}, ~~~k\in\mathcal{K},
	\vspace{-0.2cm}
\end{equation}  
%\begin{equation}\label{Eq:mp}
%	\mathbf{h}^{\rm EH}_{k}=\sqrt{N} h^{\rm EH}_{{\rm LoS}, k} \mathbf{b}(\theta_{k},r_{k})+\sqrt{\frac{N}{L_{k}}}\sum_{\ell=1}^{L_{k}}h^{\rm EH}_{{\rm NLoS}, \ell}\mathbf{b}(\theta_{\ell},r_{\ell}), ~~~k\in\mathcal{K},
%\end{equation}  
%\begin{equation}
%	\mathbf{h}^{\rm EH}_{k}=\sqrt{\frac{N}{L_{k}}}\sum_{l=1}^{L_{k}} h^{\rm EH}_{l} \mathbf{b}(\theta_{l},r_{l}), ~~~k\in\mathcal{K},
%\end{equation}  
%where $L_{k}$ denotes the number of paths between the BS and EH receiver k with $\ell=1$ and $\ell>1$ corresponding to the line-of-sight (LoS) and Non-LoS (NLoS) components, respectively,
where there exist one line-of-sight  (LoS) path and $L_{k}$  non-LoS (NLoS) paths between the BS and EH receiver $k$.
%, and $h^{\rm EH}_{{\rm LoS}, k}$ and $h^{\rm EH}_{{\rm NLoS}, \ell}$ represent the complex-valued channel gains of the LoS path and $\ell$-th NLoS path, respectively. $\mathbf{b}(\theta_{l},r_{l})$ denotes the (normalized) near-field
%channel steering vector, which will be modeled later.
In this paper, we consider SWIPT in high-frequency bands (e.g., mmWave and THz), whose channels are susceptible to blockage. Therefore, we mainly consider the LoS channel component for both EH and ID receivers, while the NLoS components are neglected due to small power \cite{9957130,liu2022low}. As such, the channel from the BS to EH receiver $k$ can be approximated as $\mathbf{h}^{\rm EH}_{k} \approx \mathbf{h}^{\rm EH}_{{\rm LoS}, k}$, which is modeled as follows. First, based on the near-field spherical wavefront model in Fig.~\ref{fig:An}, the distance between the $n$-th antenna at the BS (i.e., ($0, \delta_{n}d$)) and EH receiver $k$ is given by
%for notational brevity, we use $h^{\rm EH}_{k}$ to denote $h^{\rm EH}_{{\rm LoS}, k}$ in \eqref{Eq:mp}. Then, as shown in Fig.~\ref{fig:An}, based on the near-field spherical wavefront model, the distance between the $n$-th antenna at the BS (i.e., ($0, \delta_{n}d$)) and EH receiver $k$ is given by}
\begin{align}\label{Eq:rrr}
	r^{(n)}_{k}&=\sqrt{r_{k}^2+\delta_{n}^2d^2-2r_{k}\theta_{k} \delta_{n}d},\vspace{-0.2cm}
\end{align}
where $r_{k}$ denotes the distance between the BS antenna center and EH receiver $k$, and $\theta_{k}=2d\cos(\phi_{k})/ \lambda$ denotes the spatial angle at the BS with $\phi_{k}$ denoting the physical angle-of-departure (AoD) from the BS center to EH receiver $k$. Then, by using the second-order Taylor expansion $\sqrt{1+x}\approx1+\frac{1}{2}x-\frac{1}{8}x^2$, $r^{(n)}_{k}$ in \eqref{Eq:rrr} can be approximated as $r^{(n)}_{k}\approx r_{k}-\delta_{n}d\theta_{k}+\frac{\delta_{n}^2d^2(1-\theta^2_{k})}{2r_{k}}$, which is accurate enough when the BS and EH receiver distance $r_{k}$ is smaller than the Rayleigh distance \cite{9957130}.
%\begin{figure}[t]
%	\centering
%	\includegraphics[width=12cm]{./Antenna.pdf}
%	\caption{Illustration of the near and far-field channel models.}\label{fig:An}
%\end{figure}
As such, the LoS channel between the $n$-th BS antenna and EH receiver $k$ can be modeled as $[\mathbf{h}^{\rm EH}_{k}]_{n}=h^{\rm EH}_{k,n}e^{-j 2 \pi(r^{(n)}_{k}-r_{k})/\lambda},$
where $h^{\rm EH}_{k,n}=\frac{\lambda}{4\pi r^{(n)}_{k}}$ is the antenna-wise complex-valued channel gain of EH receiver $k$. Moreover, we assume that the EH receivers are located in the radiative  Fresnel region for which $r_{k}>r_{\rm min}=\max\{\frac{1}{2}\sqrt{\frac{D^3}{\lambda}},1.2D\}$ \cite{bjornson2021primer}. Under the above assumption, we have $h^{\rm EH}_{k,1}\approx h^{\rm EH}_{k,2}\cdots\approx h^{\rm EH}_{k,N}\triangleq h^{\rm EH}_{k}=\frac{\lambda}{4\pi r_{k}}$, where $h^{\rm EH}_{k}$ is the common complex-valued channel gain for different antennas \cite{9738442}.
%Generally, for each ${\rm U}^{\rm EH}_{k}$, the distance $r_{k}$ is large than the array aperture $D=(N-1)d$. For instance, consider a BS equipped with $256$-antenna ULA operating at $30$ GHz frequency, $r_{k}$ tends to be larger than $D=(N-1)d=255\times0.005=1.275$ m. Moreover, it has been shown in \cite{9957130} that, with $r_{k}>1.2D$, we can assume that $\frac{\lambda}{4\pi r^{(0)}_{k}}\approx\cdots\approx\frac{\lambda}{4\pi r^{(N-1)}_{k}}\approx\frac{\lambda}{4\pi r_{k}}$ based on the Fresnel approximation.

Based on the above, the LoS-dominant near-field channel from the BS to EH receiver $k$ can be simply modeled as
\begin{equation}
	\mathbf{h}^{\rm EH}_{k}\approx\sqrt{N}h^{\rm EH}_{k} \mathbf{b}(\theta_{k},r_{k}), ~~~k\in\mathcal{K},\vspace{-0.1cm}
\end{equation}  
where $\mathbf{b}(\theta_{k},r_{k})$ denotes the (normalized) near-field channel steering vector, which is given by
\begin{equation}\label{near_steering}
	\!\mathbf{b}\left(\theta_{k}, r_{k}\right)\!=\!\frac{1}{\sqrt{N}}\!\left[e^{-j 2 \pi(r^{(0)}_{k}-r_{k})/\lambda}, \cdots, e^{-j 2 \pi(r^{(N-1)}_{k}-r_{k})/\lambda}\right]^{T}.\nn	\vspace{-0.1cm}
\end{equation}

{\color{black}
\underline{\bf Far-field channel model for ID receivers:} For each ID receiver that is located in the far-field of the BS, say ID receiver $m$, its channel from the BS can be characterized as below based on the planar wavefront propagation model,
	\vspace{-0.2cm}
\begin{equation}
	\mathbf{h}^{\rm ID}_{m}=\mathbf{h}^{\rm ID}_{{\rm LoS}, m} +\sum_{\ell=1}^{L_{m}}\mathbf{h}^{\rm ID}_{{\rm NLoS}, m, \ell}, ~~~m\in\mathcal{M},	\vspace{-0.2cm}
\end{equation}
%\begin{equation}
%	\mathbf{h}^{\rm ID}_{m}=\sqrt{N} h^{\rm ID}_{{\rm LoS}, m} \mathbf{a}(\theta_{m})+\sqrt{\frac{N}{L_{m}}}\sum_{\ell=1}^{L_{m}}h^{\rm ID}_{{\rm NLoS}, \ell}\mathbf{a}(\theta_{l}), ~~~m\in\mathcal{M},
%\end{equation}
%\begin{equation}
%	\mathbf{h}^{\rm ID}_{m}=\sqrt{\frac{N}{L_{m}}}h^{\rm ID}_{l} \mathbf{a}(\theta_{l}), ~~~m\in\mathcal{M},
%\end{equation}
where there are one LoS path and $L_{m}$ NLoS paths between the BS and ID receiver $m$.
%and $h^{\rm ID}_{{\rm LoS}, m}$ and $h^{\rm ID}_{{\rm NLoS}, \ell}$ represent the complex-valued channel gains of the LoS path and $\ell$-th NLoS path, respectively.
By ignoring the negligible NLoS components in high-frequency bands, the BS$\to$ID receiver $m$ channel can be approximated by its LoS component, i.e.,
\begin{equation}
	\mathbf{h}^{\rm ID}_{m}\approx\sqrt{N}h^{\rm ID}_{m} \mathbf{a}(\theta_{m}), ~~~m\in\mathcal{M},\vspace{-0.2cm}
\end{equation}
where $h^{\rm ID}_{m}=\frac{\lambda}{4\pi r_{m}}$ represents the complex-valued channel gain of ID receiver $m$. In addition, $\mathbf{a}(\theta_{m})$ denotes the (normalized) far-field channel steering vector, given by
\begin{equation}\label{far_steering}
	\mathbf{a}(\theta_{m})\triangleq \frac{1}{\sqrt{N}}\left[1, e^{j \pi \theta_{m}},\cdots, e^{j \pi (N-1)\theta_{m}}\right]^T,\vspace{-0.2cm}
\end{equation}
where $\theta_{m}$ denotes the spatial angle at the BS with $\phi_{m}$ denoting the physical AoD from the BS center to ID receiver $m$.}
\subsubsection{Signal Model}
	Let $x^{\rm EH}_{k}$, $k\in\mathcal{K}$ denote the transmitted energy-carrying signal for EH receiver $k$ with power $P^{\rm EH}_{k}$ and $x^{\rm ID}_{m}$, $m\in\mathcal{M}$ the information-bearing signal for  ID receivers $m$  with power $P^{\rm ID}_{m}$. {Then, by applying hybrid beamforming, the transmitted signal vector by the BS is given by
	$\bar{\mathbf{x}}=\mathbf{F}_{\rm A}\mathbf{F}_{\rm D}\mathbf{x},$
where $\mathbf{x}=[x_{1}^{\rm EH},\cdots,x_{K}^{\rm EH},x_{1}^{\rm ID},\cdots,x_{M}^{\rm ID}]^{T}$, $\mathbf{F}_{\rm D}$ represents a $(K+M)\times(K+M)$ digital precoder and $\mathbf{F}_{\rm A}=[\mathbf{v}^{\rm EH}_{1},\cdots,\mathbf{v}^{\rm EH}_{K},\mathbf{v}^{\rm ID}_{1}\cdots,\mathbf{v}^{\rm ID}_{M}]$ denotes an $N\times(K+M)$ analog precoder with $\mathbf{v}^{\rm EH}_{k}$ and $\mathbf{v}^{\rm ID}_{m}$ representing the analog beamforming vector for EH receiver $k$ and ID receiver $m$, respectively.}
 In this paper, to obtain useful insights as well as reduce the hardware cost of the XL-array, we mainly consider the purely analog beamforming design to evaluate the performance gain, for which the digital precoder is set as an identity matrix, i.e., $\mathbf{F}_{\rm D}=\mathbf{1}_{(K+M),(K+M)}$.
 {\color{black}To further improve the performance, the weighted minimum mean square error (WMMSE) or zero-forcing (ZF) based digital beamforming can be properly designed given the analog precoder $\mathbf{F}_{\rm A}$. The corresponding performance of hybrid beamforming will be evaluated by simulations in Section~\ref{SE:NR}.}

Let $\mathcal{D}\triangleq  \{s^{\rm EH}_{1},\cdots,s^{\rm EH}_{K},s^{\rm ID}_{1},\cdots,s^{\rm ID}_{M}\}\in\mathbb{C}^{K+M}$ denote the beam-scheduling indicator set for the XL-array BS, where $s^{\rm EH}_{k}$ and $s^{\rm ID}_{m}$ denote respectively the binary scheduling variable for each EH receiver $k$ and ID receiver $m$. Specifically, $s^{\rm EH}_{k}=1$ if EH receiver $k$ is scheduled by the BS and $s^{\rm EH}_{k}=0$ otherwise; while $s^{\rm ID}_{m}$ is defined in a way similar to $s^{\rm EH}_{k}$.
%
%\begin{equation}
%	\mathcal{D} \triangleq  \{s^{\rm EH}_{1},\cdots,s^{\rm EH}_{K},s^{\rm ID}_{1},\cdots,s^{\rm ID}_{M}\}.  
%\end{equation}

\underline{\bf Signal model for far-field ID receivers:}
Consider the data transmission to a far-field ID receiver $m$. Its received signal is given by 
\vspace{-0.3cm}
\begin{align}
	y^{\rm ID}_{m}=(\mathbf{h}^{\rm ID}_{m})^H&\mathbf{v}^{\rm ID}_{m}s^{\rm ID}_{m}{x}^{\rm ID}_{m}+\underbrace{(\mathbf{h}^{\rm ID}_{m})^H\sum_{k=1}^{K}\mathbf{v}^{\rm EH}_{k}s^{\rm EH}_{k}{x}^{\rm EH}_{k}}_{\rm \textbf{Interference from EH signals}}\nn\\
	+&\underbrace{(\mathbf{h}^{\rm ID}_{m})^H\sum^{M}_{j=1,j\neq m}\mathbf{v}^{\rm ID}_{j}s^{\rm ID}_{j}{x}^{\rm ID}_{j}}_{\rm \textbf{Interference from other ID signals}}+z^{\rm ID}_{m},
\end{align}
where $z^{\rm ID}_{m}$ is the additive white Gaussian noise (AWGN) at ID receiver ${m}$ with zero mean and power $\sigma^2_{m}$.
As such, the received signal-to-interference-plus-noise ratio
(SINR) at ID receiver ${m}$ is given by \eqref{Eq:SNR}, as
shown at the top of the page,
\begin{figure*}[ht]
\begin{equation}\label{Eq:SNR}
	{\rm SINR}^{\rm ID}_{m}=
	\frac{s^{\rm ID}_{m}P^{\rm ID}_{m}g^{\rm ID}_{m}|\mathbf{a}^{H}\left(\theta_{m}\right)\mathbf{v}^{\rm ID}_{m}|^2}{\sum^{K}_{k=1}s^{\rm EH}_{k}P^{\rm EH}_{k}g^{\rm ID}_{m}|\mathbf{a}^{H}\left(\theta_{m}\right)\mathbf{v}^{\rm EH}_{k}|^2+\sum^{M}_{j=1,j	\neq m}s^{\rm ID}_{j}P^{\rm ID}_{j}g^{\rm ID}_{m}|\mathbf{a}^{H}\left(\theta_{m}\right)\mathbf{v}^{\rm ID}_{j}|^2+\sigma^2_{m}},\vspace{-0.1cm}
\end{equation}
\vspace{-0.5cm}
	\hrulefill
\end{figure*}
where $g^{\rm ID}_{m}=N|h^{\rm ID}_{m}|^2$. The corresponding achievable rate in bits per second per Hertz (bps/Hz) is given by $R^{\rm ID}_{m}=\log_2\left(1+{\rm SINR}^{\rm ID}_{m}\right).$
%{\color{blue}
%%In this paper, our aim is to optimize the beam scheduling and power allocation, while it is assumed that the analog beamformer $\mathbf{F}_{\rm A}$ have been well configured towards the EUs and IUs, i.e., $\mathbf{v}^{\rm ID}_{m}=\mathbf{a}\left(\theta_{m}\right)$, $\forall m \in \mathcal{M}$ and $\mathbf{v}^{\rm EH}_{k}=\mathbf{b}\left(\theta_{k}, r_{k}\right)$, $\forall k \in \mathcal{K}$.
%}

\underline{\bf Signal model for near-field EH receivers:}
For wireless power transfer (WPT), due to the broadcast property of wireless channels, each EH receiver can harvest wireless energy from both the energy and information signals. As a result, by ignoring the negligible noise power at the
EH receivers and assuming the linear EH model \cite{6860253,9110849}, the harvested power at EH receiver $k$ is given by
\begin{align}\label{Eq:HP111}
	&Q_{k}\!=\!\zeta(s^{\rm EH}_{k}|(\mathbf{h}^{\rm EH}_{k})^H\mathbf{v}^{\rm EH}_{k}{x}^{\rm EH}_{k}|^2\!\nn\\
	&\!+\!\!\underbrace{\sum_{i=1,i	\neq k}^{K}\!\!s^{\rm EH}_{i}|(\mathbf{h}^{\rm EH}_{k})^H\mathbf{v}^{\rm EH}_{i}{x}^{\rm EH}_{i}|^2}_{\rm \textbf{Harvested power from other EH signals }}\!
	+\!\underbrace{\sum_{m=1}^{M}s^{\rm ID}_{m}|(\mathbf{h}^{\rm EH}_{k})^H\mathbf{v}^{\rm ID}_{m}{x}^{\rm ID}_{m}|^2}_{\rm \textbf{Harvested power from ID signals}}),\nn
\end{align}
where $0<\zeta \leq 1$ denotes the energy harvesting efficiency.
	\vspace{-0.2cm}
\subsection{Problem Formulation}
%\begin{figure*}
%	\begin{equation}
%		\!\!\!\!\!\!\!R^{\rm ID}_{m}\!\!\left(\mathcal{D},\{P^{\rm EH}_{k}\},\{P^{\rm ID}_{m}\}\right)\!=
%		\!\log_2\!\left(\!1\!+\!\!\frac{s^{\rm ID}_{m}P^{\rm ID}_{m}g^{\rm ID}_{m}}{\sum^{K}_{k=1}s^{\rm EH}_{k}P^{\rm EH}_{k}g^{\rm ID}_{m}|\mathbf{a}^{H}\!\!\left(\theta_{m}\right)\!\mathbf{b}\!\left(\theta_{k}, r_{k}\right)\!|^2\!+\!\sum^{M}_{j=1,j	\neq m}s^{\rm ID}_{j}P^{\rm ID}_{j}g^{\rm ID}_{m}|\mathbf{a}^{H}\!\!\left(\theta_{m}\right)\!\mathbf{a}\!\left(\theta_{j}\right)\!|^2\!+\!\sigma^2_{m}}\!\right)\!\!,\label{Eq:sum_rate}
%	\end{equation}
%	\hrulefill
%	\vspace{-0.4cm}
%\end{figure*}
%\begin{figure*}
%	\begin{equation}
%		\!\!\!\!\!\!\!\!\!\! Q_{k}\!\left(\mathcal{D},\!\{P^{\rm EH}_{k}\},\!\{P^{\rm ID}_{m}\}\right)\!=\!
%		\zeta\!\!\left(\!s^{\rm EH}_{k}P^{\rm EH}_{k}g^{\rm EH}_{k}\!\!+\!\!\!\sum_{i=1,i	\neq k}^{K} \!\! s^{\rm EH}_{i}P^{\rm EH}_{i}g^{\rm EH}_{k}|(\mathbf{b}^H\!(\theta_{k},r_{k})\mathbf{b}(\theta_{i},r_{i})|^2 \!+\!\!\sum_{m=1}^{M}\!s^{\rm ID}_{m}P^{\rm ID}_{m}g^{\rm EH}_{k}|\mathbf{b}^H\!(\theta_{k},r_{k})\mathbf{a}\!\left(\theta_{m}\right)\!|^2\!\right)\!\!,\!\!\label{Eq:objecfuncx}
%	\end{equation}
%	\hrulefill
%		\vspace{-0.4cm}
%\end{figure*}
In this paper, we assume that the BS has perfect channel state information (CSI) of all EH and ID receivers, i.e., near- and far-field channel steering vectors $\mathbf{a}$ and $\mathbf{b}$ as well as channel gains $g^{\rm EH}_{k}$ and $g^{\rm ID}_{m}$. In practice, this CSI can be efficiently obtained by using existing far-field and near-field channel estimation and beam training methods (see, e.g., \cite{8949454,9129778,9913211}). In addition, for ease of implementation, we assume that if an EH or ID receiver is scheduled, the BS will steer a beam towards it to maximize its received power based on e.g., codebook based beamforming, i.e., $\mathbf{v}^{\rm ID}_{m}=\mathbf{a}\left(\theta_{m}\right)$ and $\mathbf{v}^{\rm EH}_{k}=\mathbf{b}\left(\theta_{k}, r_{k}\right)$. 
%Under the above assumptions, the achievable rate of each ID receiver ${m}$ and the harvested power at each EH receiver $k$ are given by in \eqref{Eq:sum_rate} and \eqref{Eq:objecfuncx}, respectively, where $g^{\rm EH}_{k}=N|h^{\rm EH}_{k}|^2$.

Our objective is to jointly optimize the beam scheduling (i.e., $\mathcal{D}$) and power allocation (i.e., $\{P^{\rm EH}_{k}\}$ and $\{P^{\rm ID}_{m}\}$) of the XL-array BS for maximizing the weighted sum-power harvested at all EH receivers, subject to a sum-rate constraint for all ID receivers and a total BS transmit power constraint.
Let $\alpha_{k} \ge 0$ denote a predefined power weight for each EH receiver $k$, where a larger value of $\alpha_{k}$ indicates higher preference for transferring energy to EH receiver $k$, as compared to other EH receivers. As such, the weighted sum-power transferred to all EH receivers, denoted by $Q$, can be expressed as
	\vspace{-0.2cm}
\begin{equation}\label{Eq:sum-power}
	Q\left(\mathcal{D},\{P^{\rm EH}_{k}\},\{P^{\rm ID}_{m}\}\right) =\sum_{k=1}^{K}\alpha_{k}	Q_{k}\left(\mathcal{D},\{P^{\rm EH}_{k}\},\{P^{\rm ID}_{m}\}\right).\nn
	\vspace{-0.2cm}
\end{equation}
Based on the above, this optimization problem can be formulated as follows
 \begin{subequations}
 	\begin{align}
 	\max_{\substack{\mathcal{D},\{P^{\rm EH}_{k}\},\{P^{\rm ID}_{m}\}} }  &~	Q\left(\mathcal{D},\{P^{\rm EH}_{k}\},\{P^{\rm ID}_{m}\}\right)
 		\label{Eq:Or_OJ}\\
 	\text{s.t.}~~~~~
 		&~ 	\sum_{m=1}^{M}R^{\rm ID}_{m}\left(\mathcal{D},\{P^{\rm EH}_{k}\},\{P^{\rm ID}_{m}\}\right)\ge R,\label{C:sum-rate}
 		\\	({\bf P1}):	~~~~
 		&~
 		s^{\rm EH}_{k},s^{\rm ID}_{m} \in \{0,1\},~~ k\in \mathcal{K},m\in \mathcal{M},\label{C:binary}\\
 		&~
 		 \sum_{k=1}^{K}s^{\rm EH}_{k}P^{\rm EH}_{k}+\sum_{m=1}^{M}s^{\rm ID}_{m}P^{\rm ID}_{m}\le P_{0},\label{C:Overall_P}\\
 		 &~
 		 P^{\rm EH}_{k} \ge 0, P^{\rm ID}_{m}\ge 0,~~ k\in \mathcal{K},m\in \mathcal{M}.\label{C:nonnegative}
 	\end{align}
 \end{subequations}

Problem (P1) is a mixed-integer optimization problem due to the binary beam-scheduling optimization variables $\mathcal{D}$ in \eqref{C:binary} and the continuous power-allocation variables $\{P^{\rm EH}_{k}\}$, $\{P^{\rm ID}_{m}\}$ in \eqref{C:Overall_P}. Moreover, the beam scheduling and power allocation optimization are strongly coupled in the objective function and the constraints in \eqref{C:sum-rate}--\eqref{C:Overall_P}, which renders problem (P1) more difficult to be optimally solved. To tackle these difficulties, we propose an efficient algorithm in this paper to obtain a high-quality solution to problem (P1).
	\vspace{-0.2cm}
\section{Problem Reformulation}\label{PRF}
In this section, we reformulate problem (P1) into a more compact form to facilitate the subsequent algorithm design. 
	\vspace{-0.1cm}
\subsection{Eliminating the Binary Optimization Variables}
One of the main challenges in solving problem (P1) arises from the intrinsic coupling between the binary scheduling variables (i.e., $s^{\rm EH}_{k}$ and $s^{\rm ID}_{m}$) and power allocation (i.e., $\{P^{\rm EH}_{k}\}$ and $\{P^{\rm ID}_{m}\}$), which appear in both the objective function and constraints. To address this issue, we introduce the following variables to eliminate the binary optimization variables
\begin{equation}\label{b_s_e}
	\tilde{P}^{\rm EH}_{k} = s^{\rm EH}_{k}P^{\rm EH}_{k}, ~~~~ \tilde{P}^{\rm ID}_{m} = s^{\rm ID}_{m}P^{\rm ID}_{m}.
\end{equation}
As such, the constraint \eqref{C:sum-rate} can be equivalently transformed into the following form
	\vspace{-0.2cm}
\begin{equation}\label{Eq:C:sum-rate_t}
\!\!\sum_{m=1}^{M}R^{\rm ID}_{m}(\mathcal{D},\{P^{\rm EH}_{k}\},\{P^{\rm ID}_{m}\})\!=\!\sum_{m=1}^{M}\tilde{R}^{\rm ID}_{m}(\{\tilde{P}^{\rm EH}_{k}\},\{\tilde{P}^{\rm ID}_{m}\}),
	\vspace{-0.2cm}
\end{equation}
where $\tilde{R}^{\rm ID}_{m}(\{\tilde{P}^{\rm EH}_{k}\},\{\tilde{P}^{\rm ID}_{m}\})$ is given by \eqref{NewR}, as shown at the top of next page.
 \begin{figure*}
\begin{equation}\label{NewR}
	\!\!\!\!\tilde{R}^{\rm ID}_{m}(\{\tilde{P}^{\rm EH}_{k}\},\{\tilde{P}^{\rm ID}_{m}\})=
\log_2\left(\!1\!+\!\frac{\tilde{P}^{\rm ID}_{m}g^{\rm ID}_{m}}{\sum^{K}_{k=1}\tilde{P}^{\rm EH}_{k}g^{\rm ID}_{m}|\mathbf{a}^{H}\left(\theta_{m}\right)\mathbf{b}\left(\theta_{k}, r_{k}\right)|^2\!+\!\sum^{M}_{j=1,j	\neq m}\tilde{P}^{\rm ID}_{j}g^{\rm ID}_{m}|\mathbf{a}^{H}\left(\theta_{m}\right)\mathbf{a}\left(\theta_{j}\right)|^2\!+\!\sigma^2_{m}}\!\right)\!.\vspace{-0.1cm}
\end{equation}
\hrulefill
	\vspace{-0.4cm}
 \end{figure*}
Note that there is a one-to-one correspondence between $\tilde{P}^{\rm EH}_{k}$ and $\{s^{\rm EH}_k, {P}^{\rm EH}_{k}\}$. For example, if $\tilde{P}^{\rm EH}_{k}>0$, we have $s^{\rm EH}_k=1$ and ${P}^{\rm EH}_{k}= \tilde{P}^{\rm EH}_{k}$. Otherwise, if $\tilde{P}^{\rm EH}_{k}=0$, we have $s^{\rm EH}_k=0$ and ${P}^{\rm EH}_{k}=0$. 
%It is clearly observed that, when $s^{\rm ID}_{m}=0$, the left and right hands of \eqref{Eq:C:sum-rate_t} equal to zero; while for the case $s^{\rm ID}_{m}=1$, $\tilde{P}^{\rm ID}_{m}=P^{\rm ID}_{m}$, both of which make  the left and right hands of \eqref{Eq:C:sum-rate_t} equivalent
As such, the objective function in \eqref{Eq:Or_OJ} can be expressed in a simpler form, which is given by  \eqref{NewQ}, as shown at the top of next page.
 \begin{figure*}
\begin{equation}\label{NewQ}
	\!\!\!\!\tilde{Q}(\{\tilde{P}^{\rm EH}_{k}\},\{\tilde{P}^{\rm ID}_{m}\}) =
\sum_{k=1}^{K}\!\alpha_{k}\zeta\!\!	\left(\tilde{P}^{\rm EH}_{k}g^{\rm EH}_{k}\!\!+\!\!\!\sum_{i=1,i	\neq k}^{K}\!\!\tilde{P}^{\rm EH}_{i}g^{\rm EH}_{k}|(\mathbf{b}^H(\theta_{k},r_{k})\mathbf{b}(\theta_{i},r_{i})|^2\!+\!\sum_{m=1}^{M}\tilde{P}^{\rm ID}_{m}g^{\rm EH}_{k}|\mathbf{b}^H(\theta_{k},r_{k})\mathbf{a}\left(\theta_{m}\right)|^2\!\!\right)\!\!.\vspace{-0.1cm}
\end{equation}
\hrulefill
	\vspace{-0.5cm}
 \end{figure*}
Similarly, the constraint in \eqref{C:Overall_P} can be rewritten as $\sum_{k=1}^{K}\tilde{P}^{\rm EH}_{k}+\sum_{m=1}^{M}\tilde{P}^{\rm ID}_{m}\le P_{0}.$

Based on the above, problem (P1) is equivalently expressed as follows
	\vspace{-0.2cm}
\begin{subequations}
	\begin{align}
	\max_{\substack{\{\tilde{P}^{\rm EH}_{k}\},\{\tilde{P}^{\rm ID}_{m}\}} }  &~~~\tilde{Q}(\{\tilde{P}^{\rm EH}_{k}\},\{\tilde{P}^{\rm ID}_{m}\}) \nn\\
		\text{s.t.}~~~
		&~ \sum_{m=1}^{M}\tilde{R}^{\rm ID}_{m}(\{\tilde{P}^{\rm EH}_{k}\},\{\tilde{P}^{\rm ID}_{m}\})\ge R,
		\\	({\bf P2}):~~~~
		&~
		\sum_{k=1}^{K}\tilde{P}^{\rm EH}_{k}+\sum_{m=1}^{M}\tilde{P}^{\rm ID}_{m}\le P_0,\\
			&~~ \tilde{P}^{\rm EH}_{k} \ge 0,\tilde{P}^{\rm ID}_{m}\ge 0,~~ k\in \mathcal{K},m\in \mathcal{M}.\label{C:nonnegative1}
			\\[-0.6cm]\nn
	\end{align}
\end{subequations}
%Note that, there still remain complicated correlation items (e.g., $|(\mathbf{b}^H(\theta_{k},r_{k})\mathbf{b}(\theta_{i},r_{i})|$) in the objective function and constraints, hindering the optimization of problem (P$2$). To this end, we in the next provide the definition of these correlations (e.g., inter-EH and EH-ID), based on which a closed-form expression is obtained by using the Fresnel functions \cite{zhang2023mixed}. 
\subsection{Correlation Evaluation between EH and ID Receivers}
Before solving problem (P2), we first study the correlation between the channels of EH and ID receivers. To this end, we first make a key definition below.
\begin{definition}\label{De:Correlation} 
	\emph{The correlation between any two near-field steering vectors is
		\vspace{-0.1cm}
		\begin{equation}
			\eta(\theta_{p},\theta_{q},r_{p},r_{q})=|\mathbf{b}^{H}(\theta_{p},r_{p})\mathbf{b}(\theta_{q},r_{q})|.\vspace{-0.1cm}
	\end{equation}}	
\end{definition}
%Definition~\ref{De:Correlation} provides a general expression of any two near-field steering vectors, which enjoys a very complicated form. As a result, this hinders the simplification and optimization of problem (P$2$). To this end, we derive a closed-form expression of function $\eta(\cdot)$ in Theorem~\ref{Th:correlation} by using Fresnel functions.
\begin{remark}\label{Fresnel}
	\emph{Our prior work \cite{zhang2023mixed} studies the correlation between the near- and far-field steering vectors and reveals that the DFT-based far-field beams may cause strong interference to the near-field user even when they locate in different spatial angles. However, when taking a different view from the EH perspective, the power leakage from the DFT-based far-field beams to the near-field user can be used for charging EH devices efficiently. Moreover, it is worth noting that the correlation between near-field steering vectors is a general version of the near-and-far correlation since the far-field channel model is an approximation of the near-field channel model, which is shown below.}
\end{remark}
\begin{lemma}\label{Th:correlation}
	\emph{The correlation between any two near-field steering vectors $	\eta(\theta_{p},\theta_{q},r_{p},r_{q})$ can be approximated as
		\begin{equation}\label{NN_correlation}
			\eta(\theta_{p},\theta_{q},r_{p},r_{q})\approx\left|\frac{\hat{C}(\beta_1,\beta_2)+j\hat{S}(\beta_1,\beta_2)}{2\beta_2}\right|,
		\end{equation}
		where 
		\begin{equation}\label{keypara}
			\beta_1=\frac{(\theta_q-\theta_p)}{\sqrt{d\left|\frac{1-\theta_p^2}{ r_p}-\frac{1-\theta_q^2}{  r_q}\right|}},~~~
			\beta_2=\frac{N}{2}\sqrt{d\left|\frac{1-\theta_p^2}{ r_p}-\frac{1-\theta_q^2}{  r_q}\right|}.\nn
		\end{equation}
	\vspace{-0.5cm}
	}	
\end{lemma}
\begin{proof}
The proof is similar to that in [Lemma 1,  \cite{chen2022hierarchical}] and hence is omitted for brevity.
\end{proof}

{\color{black}Lemma~\ref{Th:correlation} reveals the correlation of any two near-field steering vectors, which is fundamentally determined by two key parameters $\beta_1$ and $\beta_2$. Note that the correlation has a symmetry property, i.e., $\eta(\theta_{p},\theta_{q},r_{p},r_{q})=\eta(\theta_{q},\theta_{p},r_{q},r_{p})$.}
\subsection{Problem Reformulation} 
{\color{black}Based on Lemma \ref{Th:correlation}, $\tilde{R}(\{\tilde{P}^{\rm EH}_{k}\},\{\tilde{P}^{\rm ID}_{m}\})$ in \eqref{NewR} and $\tilde{Q}(\{\tilde{P}^{\rm EH}_{k}\},\{\tilde{P}^{\rm ID}_{m}\})$ in \eqref{NewQ} can be rewritten as functions of the EH/ID correlation, given by \eqref{NewNewR} and \eqref{NewNewQ}, as shown at the top of next page.}
\begin{figure*}
\begin{equation}\label{NewNewR}
\!\!\tilde{R}^{\rm ID}_{m}(\{\tilde{P}^{\rm EH}_{k}\!\},\!\{\!\tilde{P}^{\rm ID}_{m}\!\})\!=\!
\log_2\!\left(\!1+\!\frac{\tilde{P}^{\rm ID}_{m}g^{\rm ID}_{m}}{\sum^{K}_{k=1}\tilde{P}^{\rm EH}_{k}g^{\rm ID}_{m}\eta^2(\theta_{m},\theta_{k},r_{k})\!+\!\sum^{M}_{j=1,j	\neq m}\tilde{P}^{\rm ID}_{j}g^{\rm ID}_{m}\eta^2(\theta_{m},\theta_{j})\!+\!\sigma^2_{m}}\!\right)\!,\vspace{-0.1cm}
\end{equation}
\hrulefill
 \end{figure*}
 \begin{figure*}
 	\vspace{-0.4cm}
\begin{equation}\label{NewNewQ}
	\!\!\tilde{Q}(\{\tilde{P}^{\rm EH}_{k}\!\},\!\{\!\tilde{P}^{\rm ID}_{m}\!\})\!\!=\!\!
	\sum_{k=1}^{K}\!\alpha_{k}\zeta\!\!
	\left(\!\tilde{P}^{\rm EH}_{k}g^{\rm EH}_{k}\!+\!\sum_{i=1,i	\neq k}^{K}\tilde{P}^{\rm EH}_{i}g^{\rm EH}_{k}\eta^2(\theta_{k},\theta_{i},r_{k},r_{i})\!+\!\sum_{m=1}^{M}\tilde{P}^{\rm ID}_{m}g^{\rm EH}_{k}\eta^2(\theta_{m},\theta_{k},r_{k})\right).\vspace{-0.1cm}
\end{equation}
\hrulefill
\vspace{-0.6cm}
 \end{figure*}

To facilitate the analysis, we first define a correlation matrix as 
\begin{equation}
	\mathbf{\Lambda}=
	\begin{bmatrix}
		\eta^2_{1,1}&\eta^2_{1,2}&\cdots&\eta^2_{1,K+M}\\
		\eta^2_{2,1}&\eta^2_{2,2}&\cdots&\eta^2_{2,K+M}\\
		\vdots&&\ddots&\vdots\\
		\eta^2_{K+M,1}&\eta^2_{K+M,2}&\cdots&\eta^2_{K+M,K+M}\\
	\end{bmatrix},
\end{equation}
where $\eta^2_{p,q}$ denotes the correlation between the channel steering vectors of EH/ID receiver $p$ and EH/ID  receiver $q$, which is given by
\begin{equation}
	\eta^2_{p,q}=
	\begin{cases}
		|\mathbf{b}^{H}(\theta_{p},r_{p})\mathbf{b}(\theta_{q},r_{q})|^2 & ~\text{if}~1\le p,q \le K, \\
		|\mathbf{b}^{H}(\theta_{p},r_{p})\mathbf{a}(\theta_{q})|^2 & ~\text{if}~ 1\le p \le K,\\&~K+1\le q \le K+M,\\
			|\mathbf{a}^{H}(\theta_{p})\mathbf{a}(\theta_{q})|^2 & ~\text{if}~ K+1\le p,q\le K+M,\\
				1 & ~\text{if}~p=q.\nn\\
	\end{cases}
\end{equation}

%\subsection{Reformulating Problem (P$2$) into a Matrix-based Form}
{\color{black} Next, to rewrite problem (P2) in a more compact form, we further slightly change the form of the correlation matrix $\mathbf{\Lambda}$ by setting some of the
 diagonal elements to be zero (i.e., $([\mathbf{\Lambda}]_{(K+1),(K+1)}=0,\cdots,[\mathbf{\Lambda}]_{(K+M),(K+M)}=0)$), since we do not consider energy harvesting at the ID receivers. As such, the new correlation matrix is given by}
\begin{equation}
	\mathbf{\bar{\Lambda}}=
	\begin{bmatrix}
		1&\eta^2_{1,2}&\cdots&\eta^2_{1,K+M}\\
		\eta^2_{2,1}&1&\cdots&\eta^2_{2,K+M}\\
		\vdots&&\ddots&\vdots\\
		\eta^2_{K+M,1}&\eta^2_{K+M,2}&\cdots&0\\
	\end{bmatrix}.
\end{equation}
Let $\mathbf{y}=\left[\tilde{P}^{\rm EH}_{1},\cdots,\tilde{P}^{\rm EH}_{k},\cdots,\tilde{P}^{\rm EH}_{K},\tilde{P}^{\rm ID}_{1},\cdots,\tilde{P}^{\rm ID}_{m},\cdots,\tilde{P}^{\rm ID}_{M}\right]^{T}$ denote the vector including all power allocation optimization variables. Then, problem (P2) can be equivalently recast as follows
%Let  $\mathbf{s}=\left[{s}^{\rm EH}_{1},\cdots,{s}^{\rm EH}_{k},\cdots,{s}^{\rm EH}_{K},{s}^{\rm ID}_{1},\cdots,{s}^{\rm ID}_{m},\cdots,{s}^{\rm ID}_{M}\right]^{T}$ denote the beam scheduling decisions, and  $\mathbf{p}=\left[{P}^{\rm EH}_{1},\cdots,{P}^{\rm EH}_{k},\cdots,{P}^{\rm EH}_{K},{P}^{\rm ID}_{1},\cdots,{P}^{\rm ID}_{m},\cdots,{P}^{\rm ID}_{M}\right]^{T}$ denote the corresponding power allocation, and thus the variables that need to optimized is given by
%\begin{equation}\label{Eq:opt_var}
%\mathbf{y}\triangleq \mathbf{s} \circ \mathbf{p}=\left[{s}^{\rm EH}_{1}{P}^{\rm EH}_{1},\cdots,{s}^{\rm EH}_{k}{P}^{\rm EH}_{k},\cdots,{s}^{\rm EH}_{K}{P}^{\rm EH}_{K},{s}^{\rm ID}_{1}{P}^{\rm ID}_{1},\cdots,{s}^{\rm ID}_{m}{P}^{\rm ID}_{m},\cdots,{s}^{\rm ID}_{M}{P}^{\rm ID}_{M}\right]^{T}.
%\end{equation}
\begin{subequations}
	\begin{align}
		\max_{\substack{\mathbf{y}} }  &~~~~	(\mathbf{c}^{\rm EH})^{T}\mathbf{\bar{\Lambda}}\mathbf{y}\nn\\
		\text{s.t.}
		&~~~~ \sum_{m=1}^{M}\log_{2}\left( 1+\frac{(	\mathbf{c}^{\rm ID}_{m})^{T}\mathbf{y}}{(	\mathbf{c}^{\rm ID}_{m})^{T}\mathbf{\bar{\Lambda}}\mathbf{y}+\sigma^2_{m}}\right) \ge R,\label{C:sum-rate_m}
		\\({\bf P3}):~~
		&~~~~
		\mathbf{1}^{T}_{(K+M)\times 1}\mathbf{y}\le P_{0},\label{C:sum-power1}\\
		&~~~~
		\mathbf{y} \succeq  \mathbf{0}\label{C:nonn2},
	\end{align}
\end{subequations}
where $\mathbf{c}^{\rm EH}\!\!=\!\!\left[g^{\rm EH}_{1},\cdots,g^{\rm EH}_{k},\cdots,g^{\rm EH}_{K}, \mathbf{0}_{1\times M}\right]^{T}$ and $
	\mathbf{c}^{\rm ID}_{m}\!\!=\!\!\left[\mathbf{0}_{1\times K}, \mathbf{0}_{1\times (m-1)},g^{\rm ID}_{m}, \mathbf{0}_{(m+1)\times M}\right]^{T} ,\forall m\in \mathcal{M}.$
	
{\color{black} Compared with problem (P2), 
	problem (P3) renders a more concise form for the joint beam scheduling and power allocation optimization, since the objective function and constraints in \eqref{C:sum-power1} and \eqref{C:nonn2} are all affine. However, the sum-rate constraint in \eqref{C:sum-rate_m} is not in a convex form due to the complicated ratio terms as well as the intrinsic coupling between the power allocation. This problem will be solved in the next section.}
	\vspace{-0.1cm}
\section{Proposed Algorithm for Solving (P3)}\label{SCandGC}
	\vspace{-0.1cm}
In this section, an efficient algorithm is proposed to obtain a subopitmal solution to problem (P3). Specifically, we first introduce slack variables $\{S_{m}\}$ and $\{I_{m}\}$ such that 
\begin{equation}
	\frac{1}{S_{m}}=(\mathbf{c}^{\rm ID}_{m})^{T}\mathbf{y},~~
	I_{m}=(	\mathbf{c}^{\rm ID}_{m})^{T}\mathbf{\bar{\Lambda}}\mathbf{y}+\sigma^2_{m},~~\forall m \in \mathcal{M}.
\end{equation}
Then, problem (P3) can be equivalently transformed into the following form
\begin{subequations}
	\begin{align}
		\max_{\substack{\mathbf{y},\{S_{m}\},\{I_{m}\}} }  &~~	(\mathbf{c}^{\rm EH})^{T}\mathbf{\bar{\Lambda}}\mathbf{y}\nn\\
		\text{s.t.}~~~
		&~~~~ \sum_{m=1}^{M}\log_{2}\left( 1+\frac{1}{S_{m}	I_{m}}\right) \ge R,\label{C:71}
		\\({\bf P4}):~~
		&~~~~ 	\frac{1}{S_{m}}\leq(\mathbf{c}^{\rm ID}_{m})^{T}\mathbf{y}, ~~~\forall m \in \mathcal{M},\label{C:xx2} \\
		&~~~~ 	I_{m}\geq(	\mathbf{c}^{\rm ID}_{m})^{T}\mathbf{\bar{\Lambda}}\mathbf{y}+\sigma^2_{m}, ~~~\forall m \in \mathcal{M}, \label{C:sum-rate_m1}\\
		&~~~~
		\eqref{C:sum-power1},\eqref{C:nonn2}.\nn	\\[-0.6cm]\nn
	\end{align}
\end{subequations}
{Note that the remaining challenge for solving problem (P4) is the constraint in \eqref{C:71}. To tackle this issue, we first present an useful lemma below.}
\begin{lemma}\label{Le:sca}
	\emph{Given $x>0$ and $y>0$, $f(x,y)=\log_{2}(1+\frac{1}{xy})$ is a convex function with respect to $x$ and $y$.}
\end{lemma}

\begin{proof}
	The proof is similar to that of [Lemma 1, \cite{9139273}] and hence is omitted for brevity.
\end{proof}
{\color{black}Based on Lemma~\ref{Le:sca}, $\log_{2}(1+\frac{1}{S_{m}	I_{m}})$ in \eqref{C:71} is a joint convex function with respect to $S_{m}$ and $I_{m}$. This thus allows us to apply the SCA method to tackle the constraint \eqref{C:71} as below.}
\begin{lemma}
		\emph{By applying the first-order Taylor expansion, $\log_{2}(1+\frac{1}{S_{m}I_{m}})$ can be lower-bounded as 
		\begin{align}
			\log_{2}&\left( 1+\frac{1}{S_{m}	I_{m}}\right) \ge R^{\rm low}_{m}\triangleq \log_{2}\left(1+\frac{1}{\widetilde{S}_{m}\widetilde{I}_{m}}\right)\nn\\
			&-\frac{(S_{m}-\widetilde{S}_{m})(\log_{2}e)}{\widetilde{S}_{m}+\widetilde{S}_{m}^2\widetilde{I}_{m}}-\frac{(I_{m}-\widetilde{I}_{m})(\log_{2}e)}{\widetilde{I}_{m}+\widetilde{I}_{m}^2\widetilde{S}_{m}},\nn
		\end{align}}
	\end{lemma}
where $\{\widetilde{S}_{m},\widetilde{I}_{m}\}$ denote any given feasible points. As such, problem (P4) is approximated as the following problem
\begin{subequations}
	\begin{align}
		({\bf P5}):~~\max_{\substack{\mathbf{y},\{S_{m}\},\{I_{m}\}} }  &~~~~	(\mathbf{c}^{\rm EH})^{T}\mathbf{\bar{\Lambda}}\mathbf{y}\nn\\
		\text{s.t.}~~~
		&~~~~ \sum_{m=1}^{M}R^{\rm low}_{m}\ge R,
		\\
		&~~~~ \eqref{C:sum-power1},\eqref{C:nonn2},\eqref{C:xx2},\eqref{C:sum-rate_m1}.	\nn
					\\[-0.6cm]\nn
	\end{align}
\end{subequations}
Problem (P5) now is a convex optimization problem, which thus can be
efficiently solved via standard solvers such as CVX. {Note that the solution obtained from problem (P5) requires an iterative process until the fractional increase of the objective function is below a threshold $\xi>0$, 
based on which the optimal beam scheduling and power allocation can be easily recovered.}

\begin{remark}[Convergence and complexity analysis]
	\emph{First, consider the convergence of the proposed SCA-based algorithm. As problem (P5) is solved optimally in each iteration, the objective value of problem (P5) is non-decreasing. Moreover, since the system sum-power is upper-bounded by a finite value, the proposed algorithm is guaranteed to converge. Next, the overall complexity of the proposed SCA-based algorithm is determined by the number of iterations, denoted by $I_{\rm iter}$, and the number of optimization variables (i.e., $K+M$). For each iteration, the computational complexity for solving problem (P5) by the interior-point method can be characterized as $\mathcal{O}((K+M)^{3.5})$. As a result, the total computational complexity of the proposed SCA-based algorithm is $\mathcal{O}(I_{\rm iter}(K+M)^{3.5})$. }
\end{remark}

	\vspace{-0.3cm}
\section{Numerical Results}\label{SE:NR}
	\vspace{-0.1cm}
In this section, we present numerical results to demonstrate the effectiveness of the proposed scheme.
The considered system model is illustrated in Fig.~\ref{fig:simlu}, where one BS with $N=256$ antennas serves three EH receivers and two ID receivers. Unless specified otherwise, we set $P_0 =30$ dBm, $f=30$ GHz, $\beta=(\lambda/4\pi)^2=-62$ dB, $\sigma_{m}^2=-80$ dBm, $\zeta=50\%$, $Z=\frac{2D^2}{\lambda}=327.68$ m, $\xi=0.001$, $L_{k}=L_{m}=2, \forall k,m$ and $\alpha_{k}=1, \forall k \in \mathcal{K}$. Specifically, under the polar coordinate, the EH receivers are located at $(0,0.015Z)$, $(0.1,0.02Z)$ and $(-0.05,0.03Z)$, while the two ID receivers are located  at $(0,1.05Z)$ and $(0.05,1.2Z)$.
%\begin{table}[t]
%	\caption{{Simulation parameters}}
%	\label{Table1}
%	\centering
%%	\scalebox{0.8}{
%	\begin{tabular}{|c|c|c|c|}
%		\hline
%		{\bf{Parameter}} & {\bf{Value}} & {\bf{Parameter}} & {\bf{Value}} \\
%		\hline
%		Number of BS antennas & $N=256$ &Number of EH receivers & $K=3$ \\
%		\hline
%		Number of ID receivers & $M=2$ & Carrier frequency & $f=30$ GHz \\
%		\hline
%		Reference path-loss & $\beta=(\lambda/4\pi)^2=-62$ dB & Maximum transmit power  & $P_0 =30$ dBm \\
%		\hline
%		Noise power & $\sigma_{m}^2=-80$ dBm, $m\in\mathcal{M}$ & Energy harvest efficiency & $\zeta=50\% $ \cite{6860253}\\
%		\hline	
%		Rayleigh distance & $Z=\frac{2D^2}{\lambda}=327.68$ m &Power weight & $\alpha_{k}=1, \forall k \in \mathcal{K}$ \cite{8941080,9110849}\\ 
%		\hline	
%		Convergence threshold & $\xi=0.001$ &Number of NLoS paths & $L_{k}=L_{m}=2, \forall k,m$\\ 
%		\hline
%	\end{tabular}
%\end{table}

Moreover, for performance comparison, we consider the following four benchmark schemes:
\emph{1) Exhaustive-search scheme:} This scheme enumerates  all possible beam scheduling combinations with optimized power allocation, and then selects the best one that achieves the maximum harvested sum-power;
\emph{2) Greedy scheduling + optimized power allocation (GS+OPA) scheme}, for which the EH receiver with the highest EH priority and the ID receiver with the best channel condition are scheduled with optimized power allocation;
\emph{3) Optimal scheduling + equal power allocation (OS+EPA) scheme}, which adopts the optimal beam scheduling and allocates equal power allocation for all selected receivers;
\emph{4) All scheduling + equal power allocation (AS+EPA) scheme}, for which all EH and ID receivers are scheduled with equal power allocation.
%\begin{figure}[t]
%	\centering
%	\includegraphics[width=8cm]{./convergence.eps}
%	\caption{Harvested energy v.s. iteration number.}\label{fig:conver}
%\end{figure}

\begin{figure*}[t]
	\centering
	\subfigure[Simulation setup.]{\includegraphics[width=4.5cm]{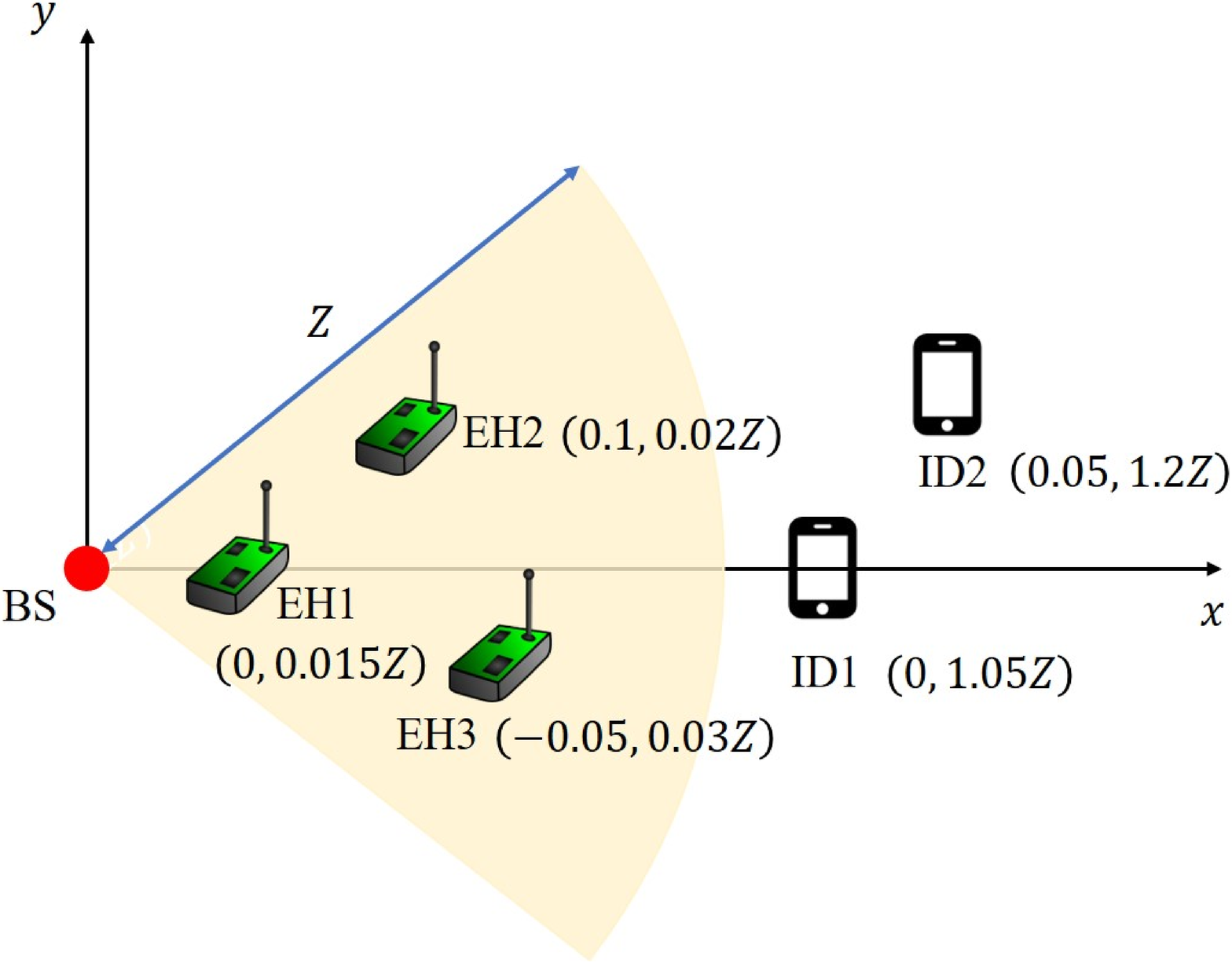}\label{fig:simlu}}
	\subfigure[$K=3$, $M=2$.]{\includegraphics[width=4.5cm]{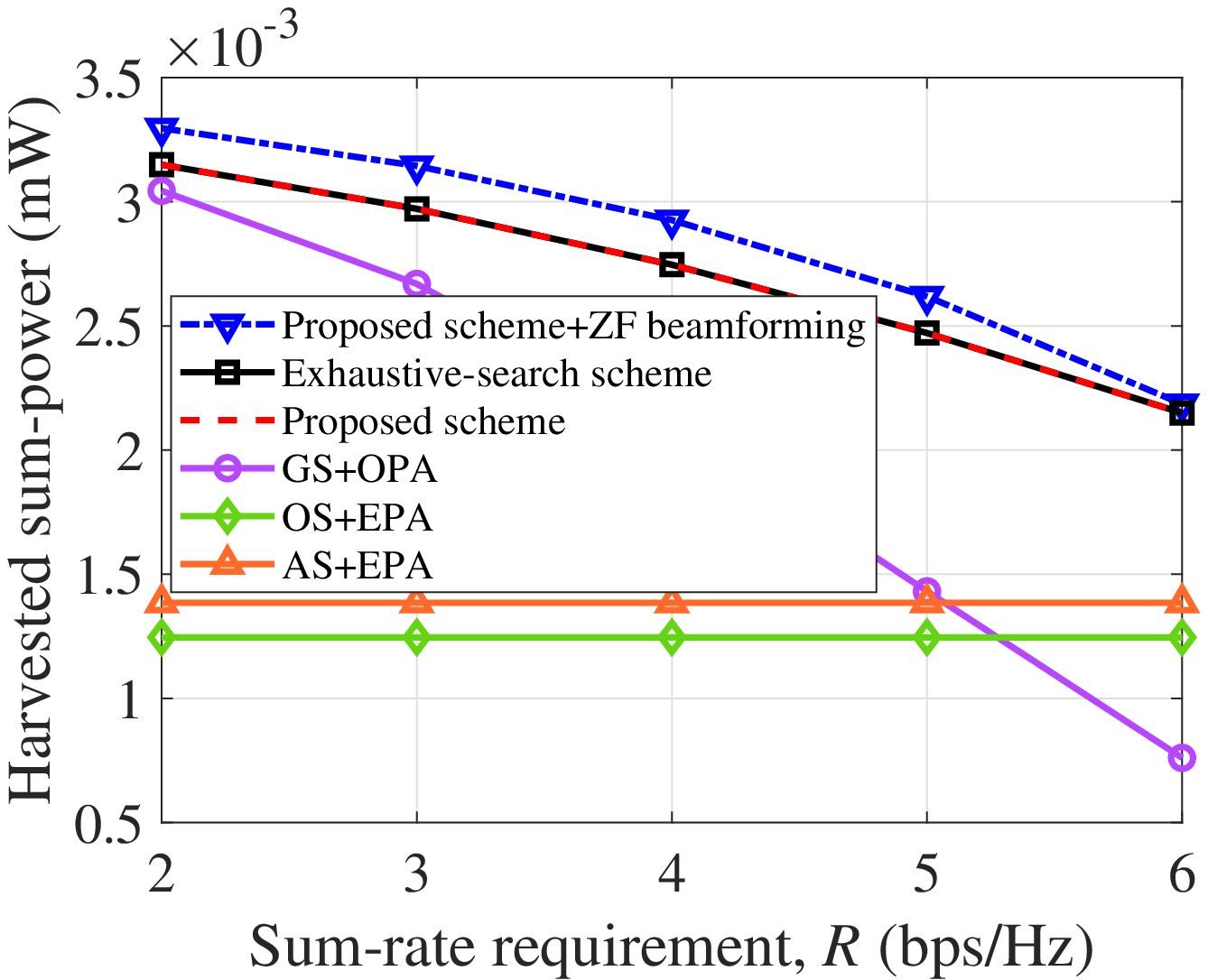}\label{fig:ratecons}}
	\subfigure[$K=3$, $M=2$.]{\includegraphics[width=4.5cm]{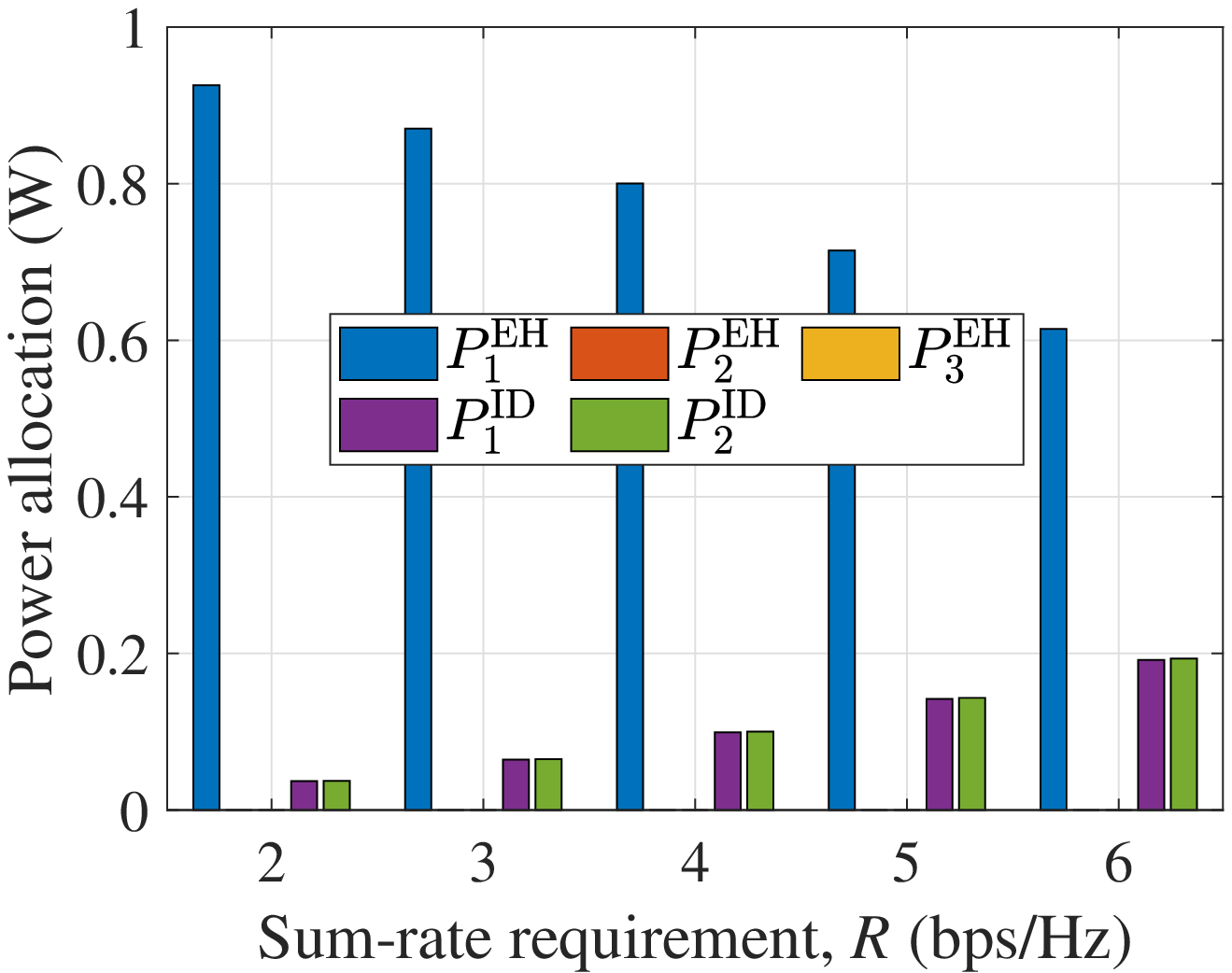}\label{fig:powall}}
	\subfigure[$K=3$, $R=5$.]{\includegraphics[width=4.5cm]{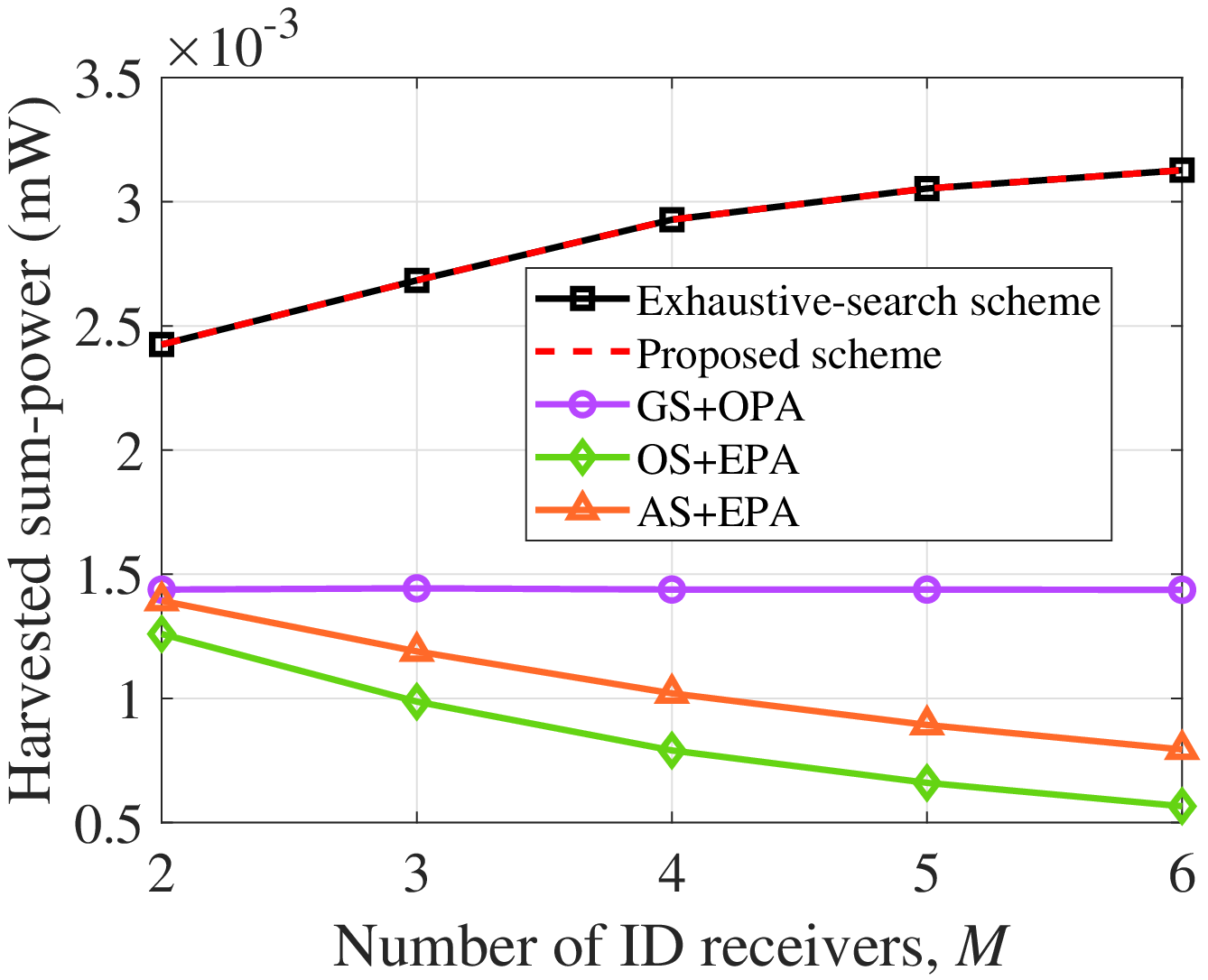}\label{fig:numIUs}}	\vspace{-6pt}
	\caption{Simulation setup and harvested sum-power versus system parameters.}
	\vspace{-18pt}
\end{figure*} 
\vspace{-0.1cm}
\subsection{Effect of Sum-Rate Requirement}
In Fig.~\ref{fig:ratecons}, we compare the harvested sum-power by different schemes versus the sum-rate requirement, $R$. First, it is observed that the proposed scheme suffers a small performance loss compared to the proposed scheme combined with the ZF digital beamforming.
Second,
for the proposed scheme and the GS+OPA scheme, the harvested power  monotonically decreases as the sum-rate requirement increases. This is because when the sum-rate requirement increases, more transmit power should be allocated to the ID beams for satisfying the sum-rate requirement, hence leaving less power for the EH beams. One interesting observation is that the gap between the proposed scheme and the GS+OPA scheme becomes larger with the growing sum-rate requirement. This is because for the proposed scheme, the power allocated to different ID beams follows the ``water-filling" structure, while the GS+OPA scheme allocates transmit power to one ID beam only, thus resulting in a lower EH efficiency.
{Moreover, we plot the power allocations to all EH and ID beams versus the sum-rate requirement in Fig.~\ref{fig:powall}. It is observed that some power is allocated to the ID beams to satisfy the sum-rate requirement, while a large portion of power is allocated to one EH beam to maximize the harvested sum-power. This is in
	sharp contrast to the conventional far-field SWIPT case, for which all power should be allocated to ID
	beams \cite{6860253}.}
%Moreover, it is observed that the performance gap between the convention far-field algorithm and the SCA-based algorithm is increased with the sum-rate constraint. This is intuitively expected, since the low-complexity algorithm chooses to provide power to the IU with the best channel condition, while the optimal power allocation scheme satisfying the sum-rate constraint actually follows the water-filling approach, thus resulting in an increased performance loss.
%%
%{\color{blue}Response to the comment: if we incorporate the water-filling design into the low-complexity algorithm, it is not feasible and we cannot obtain a closed-form solution. In addition, the only difference between the SCA-based algorithm and the low-complexity algorithm lies in the power allocation for IUs. The SCA-based algorithm follows the water-filling approach, while the low-complexity algorithm is the greedy policy.}
%\begin{figure}[t]
%	\centering
%	\subfigure[$K=3$, $M=2$, $R=5$.]{\includegraphics[width=8.4cm]{./power_vs_transmitpower.eps}\label{fig:transpow}}
%	\hspace{-16pt}
%	\subfigure[$K=3$, $M=2$.]{\includegraphics[width=8.4cm]{./power_vs_rate_cons_new30.eps}\label{fig:ratecons}}
%	\subfigure[$M=2$, $R=5$.]{\includegraphics[width=8.4cm]{./power_vs_numEUs_new30.eps}\label{fig:numEUs}}	\hspace{-16pt}
%	\subfigure[$K=3$, $R=5$.]{\includegraphics[width=8.4cm]{./power_vs_numIUs_new30.eps}\label{fig:numIUs}}
%	\caption{Harvested powers versus key parameters.}
%	\vspace{-16pt}
	\vspace{-0.2cm}
%\subsection{Effect of Number of EH Receivers}
%\vspace{-0.1cm}
%In Fig. \ref{fig:numEUs}, we evaluate the performance of the proposed scheme under different number of EH receivers, $K$. In addition to the three EH receivers in Fig.~\ref{fig:simlu}, the newly added EH receivers are uniformly distributed in an area with a radius range 
%$\left[ 0.015Z,0.3Z\right] $ and a spatial angle range $\left[ - \frac{\pi}{3},\frac{\pi}{3}\right] $. We plot in Fig. \ref{fig:numEUs} the harvested sum-powers by different schemes with  $M=2$, $R=5$ . First, it is observed that the proposed scheme significantly outperforms the benchmark schemes under different $K$. Second, we observe that the harvested sum-powers for all schemes increase with $K$. {This is because the scheduled EH and ID receivers can charge the newly added EH receiver in the considered SWIPT system.}

%1) for the proposed SCA-based scheme and the optimal beam scheduling scheme, the scheduled EH receiver can charge the newly added EH receiver in the considered SWIPT system due to the EH effect among EH receivers; 2) for the all-receivers-scheduled scheme, more power is allocated to the EH receivers when there are more EH receivers, thus leading to the highest increase rate of the harvested sum-power; 3) while for the greedy beam scheduling scheme, the harvested sum-power is increased with $K$ for the reason that the scheduled ID receiver can simultaneously charge multiple associated EH receivers, as shown in Example \ref{Ex:PL}.
\subsection{Effect of Number of ID Receivers}
Next, we show the effect of the number of ID receivers on the harvested sum-power by different schemes in Fig. \ref{fig:numIUs}. Apart from the original two ID receivers, the newly added ID receivers are uniformly distributed in an area with a radius range 
$\left[ 1.05Z,1.3Z\right] $ and a spatial angle range $\left[ - \frac{\pi}{3},\frac{\pi}{3}\right] $. First, it is observed that the harvested power with the proposed scheme monotonically increases with $M$. This is intuitively expected since when $M$ increases, the newly added ID receivers can also charge the selected EH receiver. Besides, the harvested sum-power by both OS+EPA and AS+EPA schemes decreases with $M$. This is because less transmit power is allocated to the scheduled EH beam with a larger $M$, thus reducing the total harvested sum-power.
%\begin{figure}[t]
%	\centering
%	\includegraphics[width=8cm]{./System_Model_2.pdf}
%	\caption{Schematic simulation setup with varying correlation between IU and EU.}\label{fig:simlu1}
%\end{figure}
%\begin{figure}[t]
%	\centering
%	\subfigure[Power allocation v.s. correlation]{\includegraphics[width=8cm]{./PA_vs_correlation.eps}\label{fig:pa_corre}}
%	\subfigure[Harvested energy v.s. correlation]{\includegraphics[width=8cm]{./power_vs_correlation.eps}\label{fig:power_corre}}
%	\caption{Performance illustration.}
%\end{figure}
%\subsection{Effect of Correlation}\label{section:corre}

	\vspace{-0.1cm}
\section{Conclusions}\label{CL}
	\vspace{-0.1cm}
In this paper, we considered a new mixed-field SWIPT system consisting of both near-field EH receivers and far-field ID receivers. Specifically, we first formulated an optimization problem to maximize the weighted sum-power harvested at all EH receivers by  jointly optimizing the BS beam scheduling and power allocation, subject to the ID sum-rate and BS transmit power constraints. We then proposed an efficient algorithm that leverages the binary variable elimination and SCA methods to obtain a suboptimal solution. Last, numerical results are presented to validate the effectiveness of the proposed scheme against various benchmark schemes.
	\vspace{-0.1cm}
%\begin{spacing}{1.15}
\bibliographystyle{IEEEtran}
\bibliography{IEEEabrv,Ref}
%\end{spacing}

\end{document}